\documentclass[11pt]{article}
\usepackage{amsmath,amssymb,amsthm,proof}
\usepackage{mathtools}
\usepackage{algorithmic}
\usepackage{algorithm}
\usepackage{iftex}
\ifPDFTeX
  \usepackage{graphicx}
\else\ifXeTeX
  \usepackage{graphicx}
\else\ifLuaTeX
  \usepackage{graphicx}
\else
  \usepackage[dvipdfmx]{graphicx}
\fi\fi\fi
\usepackage{listings}
\usepackage[margin=0.8in]{geometry}
\usepackage[]{natbib}
\usepackage{float}
\usepackage{eqnarray}

\bibpunct[:]{(}{)}{,}{a}{}{,}

\theoremstyle{plain}
\newtheorem{thm}{Theorem}
\newtheorem{lemma}{Lemma}

\theoremstyle{definition}

\theoremstyle{remark}
\newtheorem{result}{Result}

\allowdisplaybreaks

\begin{document}

\title{Minimizing Instability in Strategy-Proof Matching Mechanism Using A Linear Programming Approach\thanks{I am deeply grateful to my advisor, Michihiro Kandori, for his invaluable guidance and continuous support. I also thank Fuhito Kojima, Keita Kuwahara and Kento Hashimoto for their helpful and insightful comments.}}
\author{Tohya Sugano\thanks{Graduate School of Economics, The University of Tokyo. Email: sugano-tohya1011@g.ecc.u-tokyo.ac.jp} \\ The University of Tokyo, Japan}
\date{}
\maketitle

\begin{abstract}
We study the design of one-to-one matching mechanisms that are \emph{strategy-proof for both sides} and \emph{as stable as possible}. 
Motivated by the impossibility result of \citet{Roth1982-cl}, we formulate the mechanism design problem as a linear program that minimizes stability violations subject to exact strategy-proofness constraints. 
We consider both an average-case objective (summing violations over all preference profiles) and a worst-case objective (minimizing the maximum violation across profiles), and we show that imposing anonymity and symmetry when the number of agents in both sides are the same can be done without loss of optimality. 
Computationally, for small markets our approach yields randomized mechanisms with substantially lower stability violations than randomized sequential dictatorship (RSD); in the $3\times 3$ case the optimum reduces average instability to roughly one third of RSD. 
For deterministic mechanisms with three students and three schools, we find that any two-sided strategy-proof mechanism has at least two blocking pairs in the worst case and we provide a simple algorithm that attains this bound. 
Finally, we propose an extension to larger markets and present simulation evidence that, relative to sequential dictatorship (SD), it reduces the number of blocking pairs by about $0.25$ on average.
\end{abstract}

\section{Introduction}
Matching theory is a core area of economic design with deep connections to computer science and mathematics and with major applications such as kidney exchange and residency matching. 
In the classic one-to-one setting, \citet{Gale1962-mm} showed that a stable matching always exists and that deferred acceptance (DA) produces a stable matching while being strategy-proof for one side of the market. 
However, \citet{Roth1982-cl} proved a sharp impossibility: no mechanism can be simultaneously stable and strategy-proof for \emph{both} sides.

Given this impossibility, a natural design objective is \emph{second-best}: among mechanisms that are exactly two-sided strategy-proof, find those that are as stable as possible. 
While strategy-proof matching mechanisms have been studied extensively, the systematic design of two-sided strategy-proof mechanisms that minimize instability---especially beyond sequential dictatorship---remains limited.

This paper takes a computational optimization approach. 
We write the design problem as a linear program whose decision variables encode the matching probabilities assigned by the mechanism at each preference profile. 
The objective penalizes stability violations (measured by the fractional-stability slack), and the constraints enforce strategy-proofness on both sides. 
We additionally incorporate anonymity and symmetry, properties that are often regarded as socially important, and we show that doing so does not worsen the achievable instability.

Our contributions are as follows. 
First, we provide optimized randomized mechanisms for small markets that substantially reduce instability relative to RSD. 
Second, in the deterministic $3\times 3$ case we characterize the best possible worst-case performance under two-sided strategy-proofness and give an optimal deterministic algorithm. 
Finally, we propose a generalization to larger markets and show via simulation that it improves upon SD in terms of the number of blocking pairs.

\section{Related Work}
The one-to-one matching model was formalized by \citet{Gale1962-mm}, who established the existence of stable matchings and introduced deferred acceptance. 
\citet{Roth1982-cl} later showed that stability and two-sided strategy-proofness cannot be achieved simultaneously.

Beyond DA, a number of prominent mechanisms trade off incentives and stability/efficiency. 
In school choice, \citet{Abdulkadiroglu2003-ve} introduced top trading cycles (TTC), which is Pareto efficient but not two-sided strategy-proof in the two-sided matching setting. 
The Boston mechanism has also been studied; it is generally neither stable nor strategy-proof, though it can behave well under strong preference restrictions that are often viewed as unrealistic \citep{Kumano2013-pl}.

A separate line of work formulates stable matchings through linear programming. 
\citet{Roth1993-yk} characterized stable matchings as vertices of a polytope and related randomized stability notions to convex combinations of stable deterministic matchings. 
More recently, \citet{Ravindranath2021-ew} explored the incentive--stability frontier using neural networks, but such methods do not provide mechanisms that satisfy \emph{exact} strategy-proofness or guarantee optimality. 
Our contribution differs by providing an explicit linear programming formulation that enforces exact strategy-proofness and directly minimizes stability violations.

\section{Preliminaries}
Let $S=\{1,\dots,n\}$ denote the set of students and $C=\{1,\dots,m\}$ the set of schools. 
Let $P_S$ be the set of all linear orders on $C\cup\{\emptyset\}$, and let $P_C$ be the set of all linear orders on $S\cup\{\emptyset\}$. 
Each student $s\in S$ has a strict preference $\succ_s\in P_S$ over $C\cup\{\emptyset\}$, and each school $c\in C$ has a strict preference $\succ_c\in P_C$ over $S\cup\{\emptyset\}$. 
We write $\succ_S=(\succ_s)_{s\in S}$ for the student profile, $\succ_C=(\succ_c)_{c\in C}$ for the school profile, and $\succ=(\succ_i)_{i\in S\cup C}$ for the full preference profile. 
Let $P$ denote the set of all possible profiles. 
For any agent $i\in S\cup C$, let $\succ_{-i}$ denote the preferences of all agents except $i$.

A \textbf{randomized matching} is a matrix
\[
[r(s,c)]_{s\in S\cup\{\emptyset\},\, c\in C\cup\{\emptyset\}}
\]
satisfying:
\begin{enumerate}
  \item[(i)] $r(s,c)\ge 0$ for all $s\in S\cup\{\emptyset\}$ and $c\in C\cup\{\emptyset\}$;
  \item[(ii)] $\sum_{c\in C\cup\{\emptyset\}} r(s,c)=1$ for all $s\in S$;
  \item[(iii)] $\sum_{s\in S\cup\{\emptyset\}} r(s,c)=1$ for all $c\in C$.
\end{enumerate}
Let $M$ be the set of all randomized matchings. 
If $r(s,c)\in\{0,1\}$ for every $(s,c)\in S\times C$, we call $r$ \textbf{deterministic}.

By the Birkhoff--von Neumann theorem \citep{Birkhoff1946-wt,Von-Neumann1953-lh}, every randomized matching can be represented as a convex combination of deterministic matchings.

A function $g:P_S^n\times P_C^m\to M$ is called a \textbf{randomized matching mechanism}. 
If $g$ always returns a deterministic matching, we call it a \textbf{deterministic matching mechanism}.

\subsection{Stability}
A randomized matching $r$ is \textbf{ex-ante stable} if there does not exist a pair $(s,c)\in S\times C$ such that $s$ and $c$ would both prefer to be matched with each other, \emph{given the support of $r$}. 
Formally, $r$ is ex-ante stable if there do not exist $s\in S$, $c\in C$, and agents $s'\in S$, $c'\in C$ such that
\[
c \succ_s c',\quad s \succ_c s',\quad r(s,c')>0,\quad r(s',c)>0.
\]
In that case, we call $(s,c)$ a \textbf{blocking pair}.

A randomized matching $r$ is \textbf{ex-post stable} if it can be decomposed into stable deterministic matchings. 
Moreover, $r$ is \textbf{fractionally stable} if for every $(s,c)\in S\times C$ it holds that
\begin{equation}
  r(s,c)+\sum_{c'\in C:c'\succ_s c}r(s,c')+\sum_{s'\in S:s'\succ_c s}r(s',c) \ge 1.
  \label{eq:partial-stability}
\end{equation}
For deterministic matchings, ex-ante and ex-post stability coincide; we then simply say that the matching is \textbf{stable}. 
For randomized matchings, ex-post stability and fractional stability are equivalent \citep[see, e.g.,][]{Aziz2019-ms}.

A deterministic matching mechanism $g$ is \textbf{stable} if $g(\succ)$ is stable for every $\succ\in P$. 
Similarly, a randomized mechanism is \textbf{fractionally stable} if $g(\succ)$ satisfies \eqref{eq:partial-stability} for every $\succ\in P$.

For a given randomized matching $r$, define the \textbf{stability violation} for a pair $(s,c)$ as
\begin{equation}
  \max\Bigl\{ 1 - r(s,c) - \sum_{c'\in C: c'\succ_s c}r(s,c') - \sum_{s'\in S:s'\succ_c s}r(s',c),\,0 \Bigr\},
  \label{eq:stab-violation0}
\end{equation}
and the stability violation for a matching as
\begin{equation}
    \sum_{(s,c)\in S\times C}\max\Bigl\{ 1 - r(s,c) - \sum_{c'\in C: c'\succ_s c}r(s,c') - \sum_{s'\in S:s'\succ_c s}r(s',c),\,0 \Bigr\}.
    \label{eq:stab-violation}
  \end{equation}
In the deterministic case, \eqref{eq:stab-violation} coincides with the number of blocking pairs.

\subsection{Strategy-Proofness}
We use the standard (ordinal) notion of strategy-proofness for randomized assignments based on first-order stochastic dominance. 
A randomized matching mechanism $g$ is \textbf{strategy-proof} if for every student $s\in S$, every profile $\succ\in P$, every misreport $\succ'_s\in P_S$, and every school $c\in C$ with $c\succ_s\emptyset$, it holds that
\begin{equation}
\sum_{c'\in C: c' \succeq_s c} g(\succ)(s,c') \;\ge\;
\sum_{c'\in C: c' \succeq_s c} g(\succ'_s,\succ_{-s})(s,c').
\label{eq:sp_student}
\end{equation}
An analogous condition is imposed when interchanging the roles of students and schools.

\subsection{Anonymity}
Let $\pi_S:S\to S$ and $\pi_C:C\to C$ be permutations of students and schools, and let $\Pi_S$ and $\Pi_C$ denote the sets of all such permutations. 
Given a school preference order $\succ_c$ and a permutation $\pi_S$, define $\rho_C(\succ_c,\pi_S)$ as the order obtained by renaming each student $s$ as $\pi_S(s)$. 
Similarly, given a student preference order $\succ_s$ and a permutation $\pi_C$, define $\rho_S(\succ_s,\pi_C)$.

For example, if $n=3$, $(\pi_S(1),\pi_S(2),\pi_S(3))=(3,1,2)$ and $3\succ_c 2\succ_c 1$, then writing $\succ'_c=\rho_C(\succ_c,\pi_S)$ yields $2\succ'_c 1\succ'_c 3$.

Write $\rho_C(\succ_C,\pi_S)$ for the profile obtained by applying $\rho_C(\cdot,\pi_S)$ to each school's preference, and similarly for $\rho_S(\succ_S,\pi_C)$. 
A mechanism $g$ is \textbf{anonymous} if for every $\succ\in P$, every $\pi_S\in\Pi_S$, $\pi_C\in\Pi_C$, and every $s\in S$, $c\in C$, letting
\[
\succ'_S=\rho_S(\succ_S,\pi_C)\quad\text{and}\quad \succ'_C=\rho_C(\succ_C,\pi_S),
\]
we have
\begin{equation}
g(\succ)(s,c) = g\Bigl((\succ'_{\pi_S(s)})_{s\in S},\,(\succ'_{\pi_C(c)})_{c\in C}\Bigr)\bigl(\pi_S(s),\pi_C(c)\bigr).
\label{eq:anonymity}
\end{equation}
Intuitively, renaming agents does not affect outcomes.

\subsection{Symmetry}
When $n=m$, a mechanism may also be required to treat both sides of the market symmetrically. 
Fix $\succ\in P$. 
For each $i\in\{1,\dots,n\}$, construct a profile $\succ^\leftrightarrow$ by swapping the roles of students and schools: whenever in $\succ$ student $i$ ranks school $j$ above $j'$, in $\succ^\leftrightarrow$ school $i$ ranks student $j$ above $j'$, and analogously for the other side. 
A mechanism $g$ is \textbf{symmetric} if for every $(i,j)\in\{1,\dots,n\}^2$,
\begin{equation}
g(\succ)(i,j) = g(\succ^\leftrightarrow)(j,i).
\label{eq:symmetry}
\end{equation}

\subsection{Non-wastefulness}
A randomized matching mechanism $g$ is \textbf{non-wasteful} if, for every $\succ\in P$ and every pair $(s,c)\in S\times C$ with $s\succ_c\emptyset$ and $c\succ_s\emptyset$, either the student is fully matched (i.e., $\sum_{c'\in C} g(\succ)(s,c')=1$) or the school is fully matched (i.e., $\sum_{s'\in S} g(\succ)(s',c)=1$). 
Equivalently, the mechanism does not leave mutually acceptable capacity unused.

\subsection{Individual Rationality}
A randomized matching mechanism $g$ is \textbf{individually rational} if for every $\succ\in P$ and every $(s,c)\in S\times C$, whenever $g(\succ)(s,c)>0$ it follows that $s\succ_c\emptyset$ and $c\succ_s\emptyset$.

\section{The Optimization Problem for Matching Problems}
In the remainder of the paper we assume that every agent strictly prefers being matched to any partner rather than remaining unmatched:
\[
\forall s\in S,\;\forall c\in C:\quad c\succ_s \emptyset,\quad s\succ_c \emptyset.
\]
Under this assumption, any unmatched probability mass represents ``waste''.

Our goal is to design a two-sided strategy-proof mechanism that minimizes stability violations. 
We evaluate a mechanism by the total violation \eqref{eq:stab-violation} at each profile and consider two objectives:
\begin{itemize}
  \item \textbf{Objective A (average-case):} minimize the sum of violations over all profiles $\succ\in P$;
  \item \textbf{Objective B (worst-case):} minimize the maximum violation over all profiles $\succ\in P$.
\end{itemize}
We first state the average-case formulation.

\begin{equationarray}{llr}
    \underset{g}{\text{minimize}} & \multicolumn{2}{l}{\displaystyle \sum_{\succ \in P} \sum_{(s,c)\in S\times C} 
    \max\left\{ 1 - g(\succ)(s, c) - \sum_{c' \in C : c' \succ_s c} g(\succ)(s, c') 
    - \sum_{s' \in S : s' \succ_c s} g(\succ)(s', c), 0 \right\}} \\[2mm]

    \text{subject to}
    & 0 \leq g(\succ)(s, c) \leq 1 
    & \forall \succ \in P, \forall s \in S, \forall c \in C \label{const1}\\

    & \displaystyle \sum_{s \in S} g(\succ)(s, c) \leq 1 
    & \forall \succ \in P, \forall c \in C \label{const2}\\

    & \displaystyle \sum_{c \in C} g(\succ)(s, c) \leq 1 
    & \forall \succ \in P, \forall s \in S \label{const3}\\

    & \displaystyle \sum_{c' \in C : c' \succeq_s c} g(\succ)(s, c') \geq \sum_{c' \in C : c' \succeq_s c} 
    g(\succ'_s, \succ_{-s})(s, c') 
    & \forall \succ \in P, \forall \succ'_s \in P_S, \forall c \in C  \label{const4}\\ 

    & \displaystyle \sum_{s' \in S : s' \succeq_c s} g(\succ)(s', c) \geq \sum_{s' \in S : s' \succeq_c s} 
    g(\succ'_c, \succ_{-c})(s', c) 
    & \forall \succ \in P, \forall \succ'_c \in P_C, \forall s \in S \label{const5}
  \end{equationarray}

Constraints \eqref{const1}--\eqref{const3} ensure that $g$ is a feasible (possibly wasteful) randomized matching rule, while \eqref{const4} and \eqref{const5} encode two-sided strategy-proofness. 
(If one replaces \eqref{const1} by $g(\succ)(s,c)\in\{0,1\}$, the mechanism is forced to be deterministic.)
After a standard linearization of the $\max\{\cdot,0\}$ terms, the objective becomes linear and the feasible region is compact; therefore an optimal solution exists.

For the worst-case objective we replace the objective function by
\begin{equation}
  \max_{\succ\in P}\; \sum_{(s,c)\in S\times C}
  \max\Biggl\{ 1 - g(\succ)(s, c) - \sum_{c' \in C : c' \succ_s c} g(\succ)(s, c') - \sum_{s' \in S : s' \succ_c s} g(\succ)(s', c),\; 0 \Biggr\}.
\end{equation}
This criterion evaluates the mechanism independently of any distribution over preference profiles.

A direct formulation quickly becomes large: even when $n=m=3$, the number of variables is $(3!)^6\cdot 3^2=419{,}904$. 
To reduce computation and improve interpretability, we rely on two symmetrization arguments.

\begin{thm}
\label{thm1}
For any two-sided strategy-proof randomized matching mechanism $f$, there exists a two-sided strategy-proof and anonymous randomized matching mechanism $g$ such that, for both objective functions $A$ and $B$, the objective value of $g$ is no larger than that of $f$.
\end{thm}

Take any two-sided strategy-proof mechanism $f$ and construct $g$ by averaging the outcomes of $f$ over all relabelings of students and schools (i.e., over all permutations of agent names).
Anonymity holds by construction, since the distribution over relabelings is uniform.

Two-sided strategy-proofness is preserved because every strategy-proofness constraint is a linear inequality in the assignment probabilities; averaging linear inequalities keeps them valid.
Finally, the stability-violation objective at each profile is a sum of terms of the form $\max\{\text{affine function of }g(\succ),0\}$, hence it is convex in the assignment probabilities.
Therefore, averaging over relabelings cannot increase the violation for any profile and consequently it cannot increase either the average-case objective $A$ or the worst-case objective $B$.
A formal proof is given in Appendix~\ref{app:thm1}.

\begin{thm}
\label{thm2}
Suppose that $n=m$. For any two-sided strategy-proof randomized matching mechanism $f$, there exists a two-sided strategy-proof and symmetric randomized matching mechanism $g$ whose objective value (for both $A$ and $B$) is no larger than that of $f$.
\end{thm}

Enforce symmetry by averaging $f(\succ)$ with the role-swapped outcome: given a profile $\succ$, consider the profile obtained by swapping the roles of students and schools, and transpose the resulting matching.
Define $g$ as the pointwise average of these two matchings; then $g$ satisfies the symmetry identity by construction.

Two-sided strategy-proofness remains valid because the incentive constraints are linear and are preserved under averaging.
Moreover, the stability-violation objective is convex in assignment probabilities (a sum of maxima of affine functions), so averaging cannot increase either the profile-wise violation or its maximum across profiles.
A formal proof is given in Appendix~\ref{app:thm2}.

\section{Results}
We report computational results for the balanced cases $n=m=2$ and $n=m=3$.

\subsection{The Case $n=m=2$}
With two students and two schools, there exists a deterministic mechanism that is both two-sided strategy-proof and stable. 

\begin{lstlisting}[frame=single, caption={Deterministic Matching for $n=m=2$}, label={alg1}]
Input: Preference profile $\succ$
Output: Matching $r$

Set $r(s,c)=0$ for each $(s,c) \in S \times C$.
If there exists $(s,c)\in S\times C$ such that 
    student $s$ and school $c$ are each other's top choice,
then
    set $r(s,c) \leftarrow 1$ and $r(s',c') \leftarrow 1$, where $s'$ and $c'$ are the remaining agents.
Else
    let $c \in C$ be the school most preferred by student 1;
    set $r(1,c) \leftarrow 1$ and match the remaining student and school.
End if.
Return $r$.
\end{lstlisting}

Algorithm~\ref{alg1} matches a mutually top-ranked pair whenever it exists; otherwise it gives student~1 her top choice and matches the remaining pair. 
For $n=m=2$ this coincides with student-proposing DA.

\begin{lemma}
Algorithm~\ref{alg1} is two-sided strategy-proof and stable.
\end{lemma}

\begin{proof}
If a mutually top-ranked pair exists, all outcomes are fixed and no agent can improve by misreporting. 
Otherwise, student~1 receives her top choice, so she cannot benefit from misreporting; the remaining agents are forced to match. 
A symmetric argument applies to schools. 
Stability is immediate: in either branch, any agent that does not obtain her top choice is blocked by the fact that her top choice is matched to a partner that prefers that match.
\end{proof}

\subsection{The Case $n=m=3$}
\subsubsection{Deterministic Matching Mechanisms}
When restricting attention to deterministic mechanisms, our computational search shows that a two-sided strategy-proof, non-wasteful, and anonymous deterministic mechanism does not exist for $n=m=3$.

\begin{result}
\label{result1}
For $n=m=3$, there is no deterministic matching mechanism that is simultaneously two-sided strategy-proof, non-wasteful, and anonymous.
\end{result}

Result~\ref{result1} indicates that under determinism and two-sided strategy-proofness one must sacrifice anonymity (and hence equal treatment of agents with identical labels).

We also evaluate deterministic mechanisms by worst-case instability (objective function $B$).

\begin{result}
\label{result2}
For $n=m=3$, no two-sided strategy-proof deterministic matching mechanism can attain a worst-case number of blocking pairs less than or equal to $1$.
\end{result}

Equivalently, every two-sided strategy-proof deterministic mechanism must have at least two blocking pairs in the worst case. 
The next algorithm attains this lower bound and is therefore worst-case optimal among deterministic mechanisms.

\begin{lstlisting}[frame=single, caption={Worst-case Optimal Deterministic Matching for $n=m=3$}, label={alg2}]
Input: Preference profile $\succ$
Output: Matching $r$

Set $r(s,c)=0$ for each $(s,c) \in S \times C$.
Let $c \in C$ be the school most preferred by student 1.
Set $r(1,c) \leftarrow 1$; remove student 1 and school $c$.
If there exists $(s,c)$ among the remaining agents such that 
    $c$ is preferred by student $s$ to the remaining school $c'$ and
    $s$ is preferred by school $c$ to the remaining student $s'$,
then
    set $r(s,c) \leftarrow 1$ and $r(s',c') \leftarrow 1$.
Else
    let $c \in C$ be the school most preferred by student 2 among the remaining schools;
    set $r(2,c) \leftarrow 1$ and match student 3 to the remaining school.
End if.
Return $r$.
\end{lstlisting}

\begin{lemma}
Algorithm~\ref{alg2} is two-sided strategy-proof and non-wasteful, and its worst-case number of blocking pairs is $2$.
\end{lemma}

\begin{proof}
Student~1 always receives her top choice, so she cannot benefit from misreporting. 
Conditional on the first assignment, the remaining step is a deterministic rule on the residual $2\times 2$ market and thus inherits strategy-proofness from Algorithm~\ref{alg1}. 
A symmetric argument applies to schools, so the mechanism is two-sided strategy-proof. 
Non-wastefulness is immediate because the algorithm always produces a perfect matching.

For stability, note first that student~1 (matched to her top choice) cannot be part of any blocking pair. 
Any blocking pair must therefore involve one of the remaining students and the school matched to student~1; at most two such pairs can arise. 
There exist profiles where both remaining students and the school matched to student~1 form blocking pairs, so the worst case is exactly $2$.
\end{proof}

It is useful to compare with sequential dictatorship (SD), which may generate up to four blocking pairs in the worst case when $n=m=3$. 
Moreover, Algorithm~\ref{alg2} extends naturally to larger markets (Section~\ref{general_case}). 
In our computations (Table~\ref{tab2}), Algorithm~\ref{alg2} dominates SD in both average and worst-case stability violation.
Table~\ref{tab2} compares the average and worst-case stability violations (stv) between SD and Algorithm~\ref{alg2}.

\begin{table}[H]
  \centering
  \begin{tabular}{lccc}
  \hline
  \textbf{Mechanism} & \textbf{Average stv} & \textbf{Worst-case stv}  \\ \hline
  SD & 0.6666 & 3.0000 \\ 
  Algorithm~\ref{alg2} & 0.4166 & 2.0000 \\ \hline
  \end{tabular}
  \caption{Comparison of SD and Algorithm~\ref{alg2} for $n=m=3$}
  \label{tab2}
\end{table}

\subsubsection{Randomized Matching Mechanisms}
We also solve the optimization problem for randomized mechanisms under four settings: objective $A$ or $B$, with or without the non-wastefulness constraint. 
We compare the resulting optima with randomized sequential dictatorship (RSD), as well as a symmetrized version of Algorithm~\ref{alg2} (denoted by Algorithm~\ref{alg2}$'$). 

We implement two variants of RSD: 
(1) independently randomizing an ordering on both students and schools (RSD1), and 
(2) randomizing an ordering on one side only (RSD2). 
Both RSD1 and RSD2 satisfy two-sided strategy-proofness, non-wastefulness, anonymity, and symmetry. 
Table~\ref{tab1} reports average stv, worst-case stv, and average waste (computed as $3-\sum_{s\in S}\sum_{c\in C} r(s,c)$).

\begin{table}[H]
  \centering
  \begin{tabular}{lccc}
  \hline
  \textbf{Mechanism} & \textbf{Average stv} & \textbf{Worst-case stv} & \textbf{Average Waste} \\ \hline
  RSD1 & 0.6478 & 1.3333 & 0.0000 \\ 
  RSD2 & 0.6229 & 1.3333 & 0.0000 \\
  Obj.\ $A$, no nonwastefullness constraint & 0.2286 & 0.6224 & 0.0249 \\ 
  Obj.\ $A$, with nonwastefullness constraint & 0.2348 & 0.6455 & 0.0000 \\ 
  Obj.\ $B$, no nonwastefullness constraint & 0.4218 & 0.5000 & 0.0505 \\ 
  Obj.\ $B$, with nonwastefullness constraint & 0.3730 & 0.5000 & 0.0000 \\
  Algorithm~\ref{alg2}$'$ & 0.4063 & 1.6666 & 0.0000 \\ \hline
  \end{tabular}
  \caption{Comparison of optimized randomized mechanisms ($n=m=3$)}
  \label{tab1}
\end{table}

The optimized mechanisms dramatically reduce average instability relative to RSD; under objective $A$ with the non-wastefulness constraint, average instability is roughly one third of RSD.

\subsection{Generalization to $n\ge4$}
\label{general_case}
Algorithm~\ref{alg2} can be extended to larger balanced markets by iteratively assigning the highest-priority remaining agent (from a fixed global ordering) to her top choice among the remaining agents on the other side, until two students and two schools remain; the residual $2\times 2$ problem is then resolved in a way that limits blocking pairs.

\begin{lstlisting}[frame=single, caption={Generalization of Algorithm~\ref{alg2}}, label={alg3}]
Input: Preference profile $\succ$ and a fixed ordering $(i_k)_{k=1}^{2n}$ of the agents in $S\cup C$
Output: Matching $r$

Initialize $r(s,c)=0$ for every $(s,c)\in S\times C$.
Set $k=1$.
While $|S| > 2$ do:
    Let $i=i_k$.
    If $i\notin S\cup C$, then set $k\leftarrow k+1$.
    Else if $i\in S$, then
        Let $j\in C$ be the school most preferred by $i$.
        Set $r(i,j)\leftarrow 1$, remove $i$ from $S$ and $j$ from $C$.
    Else if $i\in C$, then
        Let $j\in S$ be the student most preferred by $i$.
        Set $r(j,i)\leftarrow 1$, remove $i$ from $C$ and $j$ from $S$.
    End if.
    Set $k\leftarrow k+1$.
End while.
If there exists $(s,c)\in S\times C$ satisfying 
    "$c \succ_s c'$ and $s \succ_c s'$" (with $s\neq s'$ and $c\neq c'$)
then
    set $r(s,c)\leftarrow 1$ and $r(s',c')\leftarrow 1$.
Else
    Let $i\in S\cup C$ be the agent highest in the fixed ordering.
    If $i\in S$, then
        Let $j\in C$ be the school most preferred by $i$.
        Let $i'\neq i$ and $j'\neq j$ be the remaining agents.
        Set $r(i,j)\leftarrow 1$ and $r(i',j')\leftarrow 1$.
    Else
        (Symmetric assignment.)
    End if.
End if.
Return $r$.
\end{lstlisting}

When Algorithm~\ref{alg3} and SD use the \emph{same} underlying ordering, Algorithm~\ref{alg3} produces a matching whose number of blocking pairs is no larger than that under SD. 
Intuitively, Algorithm~\ref{alg3} prevents newly matched agents from creating blocking pairs with agents that were matched earlier, while SD can create such conflicts. 
However, once one symmetrizes a mechanism to enforce anonymity and symmetry, the comparison can depend on the symmetrization method (as suggested by Table~\ref{tab1}).

\subsection{Simulation}
We conducted simulations to evaluate the performance of Algorithm~\ref{alg3} and SD as the market size grows. 
For each $n=m\in\{2,\dots,10\}$, we generated 1000 random preference profiles and computed the average stability violation using \eqref{eq:stab-violation}. 
In our implementation both Algorithm~\ref{alg3} and SD used the natural ordering $1,2,\dots,n$ on each side.

Figure~\ref{fig1} shows that for every $n$ in this range, Algorithm~\ref{alg3} yields approximately $0.25$ fewer blocking pairs on average than SD. 
In each instance the difference is either $0$ or $1$, with probability about $0.25$ that Algorithm~\ref{alg3} improves by one blocking pair.

We also compared anonymized and symmetric randomized mechanisms using two symmetrization methods: 
(1) independently randomizing orderings on students and schools, and 
(2) randomizing an ordering on the union $S\cup C$. 
Let Algorithm~\ref{alg3}$'$ denote Algorithm~\ref{alg3} symmetrized by (1) and Algorithm~\ref{alg3}'' denote the version symmetrized by (2). 
Similarly, let RSD1 and RSD2 denote the corresponding variants of RSD. 
Figure~\ref{fig2} reports the simulation results. 
In both comparisons, method (2) yields lower average stability violations, and Algorithm~\ref{alg3} remains consistently about $0.25$ better than SD on average.

\begin{figure}[H]
    \centering
    \begin{minipage}{0.45\linewidth}
        \centering
        \includegraphics[width=\linewidth]{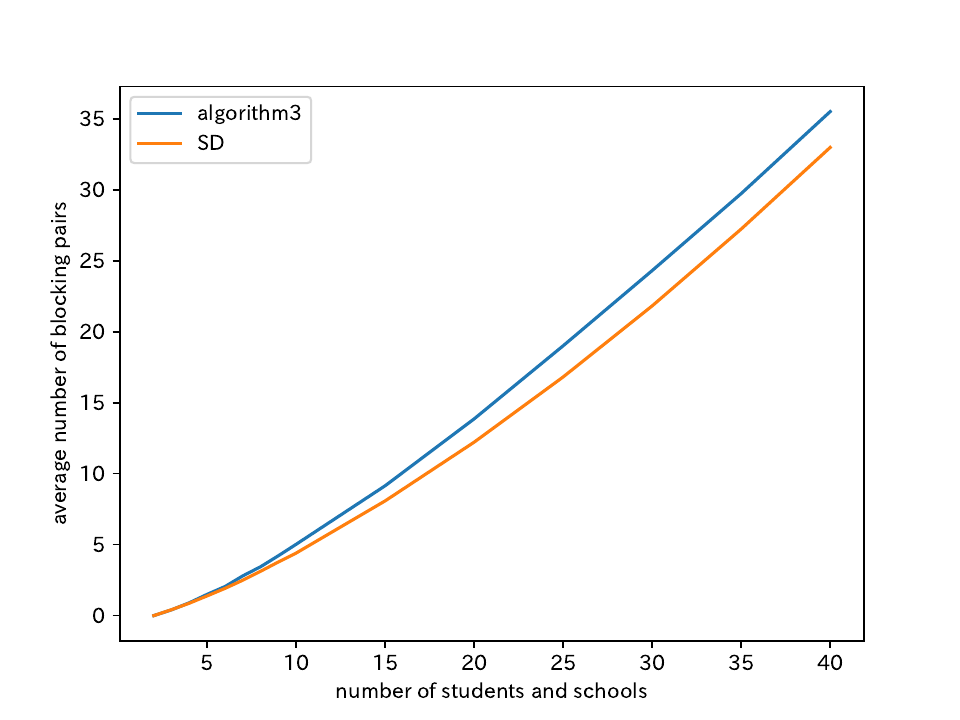}
        \caption{Comparison of Algorithm~\ref{alg3} and SD}
        \label{fig1}
    \end{minipage}
    \hspace{0.02\linewidth}
    \begin{minipage}{0.45\linewidth}
        \centering
        \includegraphics[width=\linewidth]{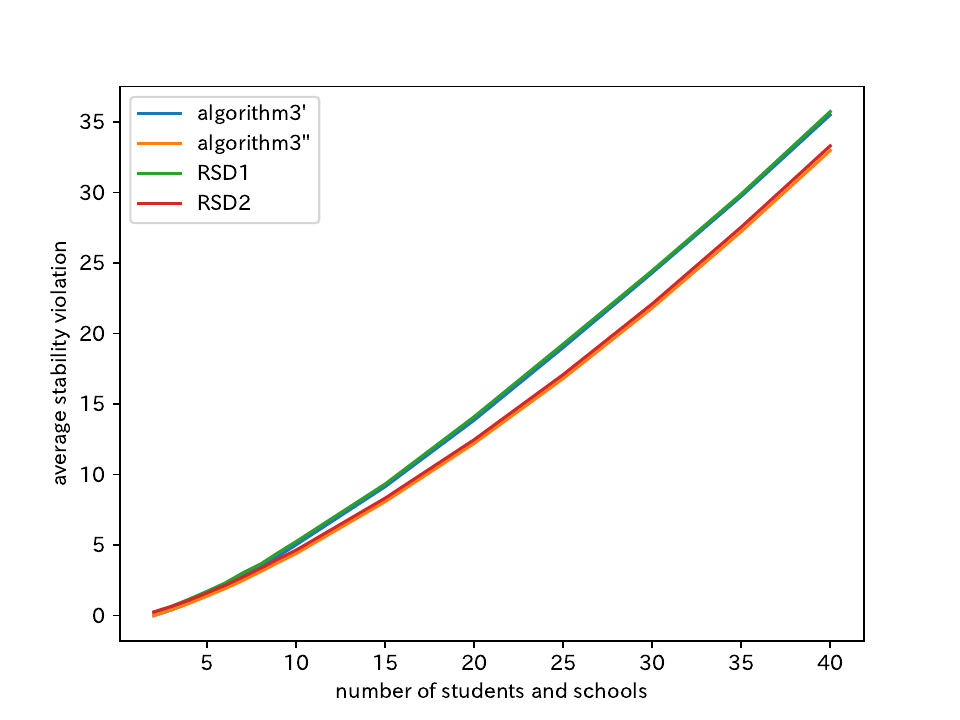}
        \caption{Comparison of anonymized and symmetric mechanisms}
        \label{fig2}
    \end{minipage}
\end{figure}

\section{Conclusion}
We formulated the design of two-sided strategy-proof matching mechanisms as a linear program that minimizes stability violations. 
The formulation allows either an average-case or a worst-case notion of instability, and it can incorporate normative requirements such as anonymity and symmetry. 
We provided symmetrization arguments showing that these requirements can be imposed without worsening achievable instability, which significantly reduces computational complexity.

For $n=m=2$, we presented a deterministic mechanism that is both two-sided strategy-proof and stable. 
For $n=m=3$, we computationally established that any two-sided strategy-proof deterministic mechanism must have at least two blocking pairs in the worst case and provided a simple algorithm achieving this bound. 
For randomized mechanisms in the $3\times 3$ market, we found mechanisms that reduce instability by roughly one third compared to RSD. 
Finally, we proposed a generalization to larger markets and, through simulation, showed that it consistently reduces blocking pairs relative to SD, by about $0.25$ on average.

More broadly, the optimization approach developed here can be applied to other small-scale matching settings (e.g., capacity-constrained matchings or roommate markets), offering a computational perspective on impossibility results by identifying mechanisms that are optimal subject to binding incentive constraints.

\appendix
\section{Appendix: Omitted Proofs}

\subsection{Proof of Theorem~\ref{thm1}}\label{app:thm1}
\begin{proof}
  Let $\succ \in P$ be given. Define $A(\succ)$ to be the set of all outcomes resulting from every possible permutation, that is,
  \[
  A(\succ) = \Bigl\{\succ' \in P : \succ' = \Bigl(\bigl(\succ_{\pi_S(s)}\bigr)_{s\in S},\, \bigl(\succ_{\pi_C(c)}\bigr)_{c\in C}\Bigr) \Bigr\}.
  \]
  Then, for any $\succ, \succ' \in P$, we have that $\succ' \in A(\succ)$ if and only if $\succ \in A(\succ')$. Moreover, if $\succ'$ does not belong to $A(\succ)$ then
  \[
  A(\succ) \cap A(\succ') = \emptyset.
  \]
  Thus, there exist profiles $\succ^1,\dots,\succ^k$ such that for any $i,j\in \{1,\dots,k\}$ with $i \neq j$, we have
  \[
  A(\succ^i) \cap A(\succ^j) = \emptyset,
  \]
  and
  \[
  \bigcup_{i=1}^k A(\succ^i) = P.
  \]
  
  Let $f$ be a strategy-proof probabilistic matching mechanism. We now define a mechanism $g$ as follows: For every $\succ \in P$, every $s \in S$, and every $c \in C$, let
  \begin{equation}
    g(\succ)(s,c) = \frac{1}{|A(\succ)|}\sum_{\pi_S \in \Pi_S,\, \pi_C \in \Pi_C} f\Bigl(\bigl(\rho_S(\succ_{\pi_S(s)},\pi_C)\bigr)_{s\in S},\, \bigl(\rho_C(\succ_{\pi_C(c)},\pi_S)\bigr)_{c\in C}\Bigr)\Bigl(\pi_S(s),\pi_C(c)\Bigr).
  \end{equation}
  In other words, $g$ returns the average matching probability over all permutations of the agents. We now show that $g$ is a strategy-proof and anonymous probabilistic matching mechanism that achieves the same objective value as $f$.
  
  \textbf{(Anonymity)}\quad
  Let $\pi_S, \pi'_S \in \Pi_S$ be given. Then there exists a unique $\pi''_S \in \Pi_S$ such that
  \[
  \pi'_S\bigl(\pi_S(s)\bigr) = \pi''_S(s).
  \]
  Similarly, for any $\pi_C, \pi'_C \in \Pi_C$, there exists a unique $\pi''_C \in \Pi_C$ satisfying
  \[
  \pi'_C\bigl(\pi_C(c)\bigr) = \pi''_C(c).
  \]
  For such $\pi_S, \pi'_S, \pi''_S, \pi_C, \pi'_C,$ and $\pi''_C$, we have
  \[
  \rho_S\Bigl(\rho_S\bigl(\succ_{\pi'_S(\pi_S(s))},\pi_C\bigr),\pi'_C\Bigr) = \rho_S\bigl(\succ_{\pi''_S(s)},\pi''_C\bigr).
  \]
  An analogous property holds when the roles of $S$ and $C$ are interchanged. Therefore, for any $\pi_S, \pi'_S \in \Pi_S$ and $\pi_C, \pi'_C \in \Pi_C$, there exists a unique pair $\pi''_S, \pi''_C$ such that
  \begin{gather*}
    f\left((\rho_S(\rho_S(\succ_{\pi'(\pi(s))},\pi_C),\pi'_C))_{s\in S}, (\rho_C(\rho_C(\succ_{\pi'(\pi(c))},\pi_S),\pi'_S))_{c\in C} \right)(\pi'_S(\pi_S(s)),\pi'_C(\pi_C(c))) \\
    = f\left(\left(\rho_S(\succ_{\pi''_S(s)},\pi''_C)\right)_{s\in S},\left(\rho_C(\succ_{\pi''_C(c)},\pi''_S)\right)_{c\in C}\right)\left(\pi''_S(s),\pi''_C(c)\right)
  \end{gather*}
  Thus, for any $\pi_S, \pi_C$, we have
  \begin{align*}
    &g\left((\rho(\succ_{\pi_S(s)},\pi_C))_{s\in S},(\rho(\succ_{\pi_C(c)},\pi_S))_{c\in C}\right)\left(\pi_S(s),\pi_C(c)\right)\\
    &=\frac{1}{|A(\succ)|}\sum_{\pi'_S\in \Pi_S,\pi'_C\in \Pi_C}f\left((\rho_S(\rho_S(\succ_{\pi'(\pi(s))},\pi_C),\pi'_C))_{s\in S}, (\rho_C(\rho_C(\succ_{\pi'(\pi(c))},\pi_S),\pi'_S))_{c\in C} \right)(\pi'_S(\pi_S(s)),\pi'_C(\pi_C(c)))\\
    &=\frac{1}{|A(\succ)|}\sum_{\pi''_S\in \Pi_S,\pi''_C\in \Pi_C}f\left(\left(\rho_S(\succ_{\pi''_S(s)},\pi''_C)\right)_{s\in S},\left(\rho_C(\succ_{\pi''_C(c)},\pi''_S)\right)_{c\in C}\right)\left(\pi''_S(s),\pi''_C(c)\right)\\
    &=g(\succ)(s,c)
  \end{align*}
  
  \textbf{(Strategy-Proofness)}\quad
  Let $\succ'_s \in P_S$ be an alternative preference for student $s$, and denote by $\succ' = (\succ'_s, \succ_{-s})$. Since $f$ is strategy-proof, for every $\succ'_s \in P_S$ and every $c \in C$ we have
  \[
  \sum_{c' \in C: c' \succeq_s c} f(\succ)(s,c') \ge \sum_{c' \in C: c' \succeq_s c} f(\succ')(s,c').
  \]
  Since this inequality remains valid under any permutation of the agents, it follows that for any $\pi_S, \pi_C$,
  \begin{gather*}
    \sum_{c'\in C: c'\rho(\succeq_{\pi_S(s)},\pi_C) \pi_C(c)}f\left(\left(\rho_S(\succ_{\pi_S(s)},\pi_C)\right)_{s\in S},\left(\rho_C(\succ_{\pi_C(c)},\pi_S)\right)_{c\in C}\right)\\
    \geq \sum_{c'\in C: c'\rho(\succeq'_{\pi_S(s)},\pi_C) \pi_C(c)}f\left(\left(\rho_S(\succ'_{\pi_S(s)},\pi_C)\right)_{s\in S},\left(\rho_C(\succ'_{\pi_C(c)},\pi_S)\right)_{c\in C}\right)(\pi_S(s),c')
  \end{gather*}
  Since $|A(\succ)| = |A(\succ')|$, it follows that for every $\succ \in P$, every $s \in S$, and every $\succ'_s \in P_S$, we have
  \begin{align*}
    &\sum_{c'\in C: c'\succeq_s c}g(\succ)(s,c')\\
    &= \frac{1}{|A(\succ)|}\sum_{c'\in C: c'\succeq_s c}\sum_{\pi_S\in \Pi_S,\pi_C\in \Pi_C}f\left(\left(\rho_S(\succ_{\pi_S(s)},\pi_C)\right)_{s\in S},\left(\rho_C(\succ_{\pi_C(c)},\pi_S)\right)_{c\in C}\right)\left(\pi_S(s),\pi_C(c')\right)\\
    &= \frac{1}{|A(\succ)|}\sum_{\pi_S\in \Pi_S,\pi_C\in \Pi_C}\sum_{c'\in C: c'\rho(\succeq_{\pi_S(s)},\pi_C) \pi_C(c)}f((\rho(\succ_s,\pi_C))_{s\in S},(\rho(\succ_c,\pi_S)))_{c\in C}(\pi_S(s),c') \\
    &\geq \frac{1}{|A(\succ')|}\sum_{\pi_S\in \Pi_S,\pi_C\in \Pi_C}\sum_{c'\in C: c'\rho(\succeq_{\pi_S(s)},\pi_C) \pi_C(c)}f((\rho(\succ'_s,\pi_C))_{s\in S},(\rho(\succ'_c,\pi_S)))_{c\in C}(\pi_S(s),c') \\
    &= \frac{1}{|A(\succ')|}\sum_{c'\in C: c'\succeq_s c}\sum_{\pi_S\in \Pi_S,\pi_C\in \Pi_C}f\left(\left(\rho_S(\succ'_{\pi_S(s)},\pi_C)\right)_{s\in S},\left(\rho_C(\succ'_{\pi_C(c)},\pi_S)\right)_{c\in C}\right)\left(\pi_S(s),\pi_C(c')\right)\\
    &= \sum_{c'\in C: c'\succeq_s c}g(\succ')(s,c')
  \end{align*}
  A similar argument holds when interchanging the roles of $s$ and $c$.
  
  \textbf{(Achievablity)}\quad
  We now show that the objective function value achievable by $f$ (with respect to objective function $A$) is also achievable by $g$. First, observe that
  \begin{align*}
    &1 - g(\succ)(s, c) - \sum_{c' \in C : c' \succ_s c} g(\succ)(s, c') - \sum_{s' \in S : s' \succ_c s} g(\succ)(s', c)\\
    &= 1- \frac{1}{|A(\succ)|}\sum_{\pi_S\in \Pi_S,\pi_C\in \Pi_C}f\left(\left(\rho_S(\succ_{\pi_S(s)},\pi_C)\right)_{s\in S},\left(\rho_C(\succ_{\pi_C(c)},\pi_S)\right)_{c\in C}\right)\left(\pi_S(s),\pi_C(c)\right) \\
    &\quad -\sum_{c' \in C : c' \succ_s c}\frac{1}{|A(\succ)|}\sum_{\pi_S\in \Pi_S,\pi_C\in \Pi_C}f\left(\left(\rho_S(\succ_{\pi_S(s)},\pi_C)\right)_{s\in S},\left(\rho_C(\succ_{\pi_C(c)},\pi_S)\right)_{c\in C}\right)\left(\pi_S(s),\pi_C(c')\right) \\
    &\quad -\sum_{s' \in S : s' \succ_c s}\frac{1}{|A(\succ)|}\sum_{\pi_S\in \Pi_S,\pi_C\in \Pi_C}f\left(\left(\rho_S(\succ_{\pi_S(s)},\pi_C)\right)_{s\in S},\left(\rho_C(\succ_{\pi_C(c)},\pi_S)\right)_{c\in C}\right)\left(\pi_S(s'),\pi_C(c)\right) \\
    &= \frac{1}{|A(\succ)|}\Biggl[\sum_{\pi_S\in \Pi_S,\pi_C\in \Pi_C}\biggl\{1- f\left(\left(\rho_S(\succ_{\pi_S(s)},\pi_C)\right)_{s\in S},\left(\rho_C(\succ_{\pi_C(c)},\pi_S)\right)_{c\in C}\right)\left(\pi_S(s),\pi_C(c)\right)\\ 
    &\quad -\sum_{c' \in C : c'\rho(\succ_{\pi_S(s)},\pi_C) \pi_C(c)}f\left(\left(\rho_S(\succ_{\pi_S(s)},\pi_C)\right)_{s\in S},\left(\rho_C(\succ_{\pi_C(c)},\pi_S)\right)_{c\in C}\right)\left(\pi_S(s),c'\right)\\
    &\quad -\sum_{s' \in S : s'\rho(\succ_{\pi_C(c)},\pi_S) \pi_S(s)}f\left(\left(\rho_S(\succ_{\pi_S(s)},\pi_C)\right)_{s\in S},\left(\rho_C(\succ_{\pi_C(c)},\pi_S)\right)_{c\in C}\right)\left(s',\pi_C(c)\right) \biggr\}\Biggr]
  \end{align*}
  Thus,
  \begin{align*}
    &\max \left\{1 - g(\succ)(s, c) - \sum_{c' \in C : c' \succ_s c} g(\succ)(s, c') - \sum_{s' \in S : s' \succ_c s} g(\succ)(s', c), 0 \right\}\\
    &= \max \Biggl\{\frac{1}{|A(\succ)|}\Biggl[\sum_{\pi_S\in \Pi_S,\pi_C\in \Pi_C}\biggl\{1- f\left(\left(\rho_S(\succ_{\pi_S(s)},\pi_C)\right)_{s\in S},\left(\rho_C(\succ_{\pi_C(c)},\pi_S)\right)_{c\in C}\right)\left(\pi_S(s),\pi_C(c)\right)\\ 
    &\quad -\sum_{c' \in C : c'\rho(\succ_{\pi_S(s)},\pi_C) \pi_C(c)}f\left(\left(\rho_S(\succ_{\pi_S(s)},\pi_C)\right)_{s\in S},\left(\rho_C(\succ_{\pi_C(c)},\pi_S)\right)_{c\in C}\right)\left(\pi_S(s),c'\right)\\
    &\quad -\sum_{s' \in S : s'\rho(\succ_{\pi_C(c)},\pi_S) \pi_S(s)}f\left(\left(\rho_S(\succ_{\pi_S(s)},\pi_C)\right)_{s\in S},\left(\rho_C(\succ_{\pi_C(c)},\pi_S)\right)_{c\in C}\right)\left(s',\pi_C(c)\right) \biggr\}\Biggr], 0 \Biggr\} \\
    &\leq \frac{1}{|A(\succ)|}\sum_{\pi_S\in \Pi_S,\pi_C\in \Pi_C} \max \Biggl\{1- f\left(\left(\rho_S(\succ_{\pi_S(s)},\pi_C)\right)_{s\in S},\left(\rho_C(\succ_{\pi_C(c)},\pi_S)\right)_{c\in C}\right)\left(\pi_S(s),\pi_C(c)\right)\\ 
    &\quad -\sum_{c' \in C : c'\rho(\succ_{\pi_S(s)},\pi_C) \pi_C(c)}f\left(\left(\rho_S(\succ_{\pi_S(s)},\pi_C)\right)_{s\in S},\left(\rho_C(\succ_{\pi_C(c)},\pi_S)\right)_{c\in C}\right)\left(\pi_S(s),c'\right)\\
    &\quad -\sum_{s' \in S : s'\rho(\succ_{\pi_C(c)},\pi_S) \pi_S(s)}f\left(\left(\rho_S(\succ_{\pi_S(s)},\pi_C)\right)_{s\in S},\left(\rho_C(\succ_{\pi_C(c)},\pi_S)\right)_{c\in C}\right)\left(s',\pi_C(c)\right), 0 \Biggr\}
  \end{align*}
  Moreover, for every $\succ \in P$, there is exactly one pair $\pi_S, \pi_C$ such that
  \[
  \succ = \Bigl(\bigl(\rho_S(\succ_{\pi_S(s)},\pi_C)\bigr)_{s\in S},\, \bigl(\rho_C(\succ_{\pi_C(c)},\pi_S)\bigr)_{c\in C}\Bigr).
  \]
  Therefore,
  \begin{align*}
    &\sum_{\succ\in P}\max \left\{1 - g(\succ)(s, c) - \sum_{c' \in C : c' \succ_s c} g(\succ)(s, c') - \sum_{s' \in S : s' \succ_c s} g(\succ)(s', c), 0 \right\}\\
    &= \sum_{\succ\in P}\max \Biggl\{\frac{1}{|A(\succ)|}\Biggl[\sum_{\pi_S\in \Pi_S,\pi_C\in \Pi_C}\biggl\{1- f\left(\left(\rho_S(\succ_{\pi_S(s)},\pi_C)\right)_{s\in S},\left(\rho_C(\succ_{\pi_C(c)},\pi_S)\right)_{c\in C}\right)\left(\pi_S(s),\pi_C(c)\right)\\ 
    &\quad -\sum_{c' \in C : c'\rho(\succ_{\pi_S(s)},\pi_C) \pi_C(c)}f\left(\left(\rho_S(\succ_{\pi_S(s)},\pi_C)\right)_{s\in S},\left(\rho_C(\succ_{\pi_C(c)},\pi_S)\right)_{c\in C}\right)\left(\pi_S(s),c'\right)\\
    &\quad -\sum_{s' \in S : s'\rho(\succ_{\pi_C(c)},\pi_S) \pi_S(s)}f\left(\left(\rho_S(\succ_{\pi_S(s)},\pi_C)\right)_{s\in S},\left(\rho_C(\succ_{\pi_C(c)},\pi_S)\right)_{c\in C}\right)\left(s',\pi_C(c)\right) \biggr\}\Biggr], 0 \Biggr\} \\
    &\leq \frac{1}{|A(\succ)|}\sum_{\succ\in P}\sum_{\pi_S\in \Pi_S,\pi_C\in \Pi_C} \max \Biggl\{1- f\left(\left(\rho_S(\succ_{\pi_S(s)},\pi_C)\right)_{s\in S},\left(\rho_C(\succ_{\pi_C(c)},\pi_S)\right)_{c\in C}\right)\left(\pi_S(s),\pi_C(c)\right)\\ 
    &\quad -\sum_{c' \in C : c'\rho(\succ_{\pi_S(s)},\pi_C) \pi_C(c)}f\left(\left(\rho_S(\succ_{\pi_S(s)},\pi_C)\right)_{s\in S},\left(\rho_C(\succ_{\pi_C(c)},\pi_S)\right)_{c\in C}\right)\left(\pi_S(s),c'\right)\\
    &\quad -\sum_{s' \in S : s'\rho(\succ_{\pi_C(c)},\pi_S) \pi_S(s)}f\left(\left(\rho_S(\succ_{\pi_S(s)},\pi_C)\right)_{s\in S},\left(\rho_C(\succ_{\pi_C(c)},\pi_S)\right)_{c\in C}\right)\left(s',\pi_C(c)\right), 0 \Biggr\} \\
    &= \sum_{\succ\in P}\max \left\{1 - f(\succ)(s, c) - \sum_{c' \in C : c' \succ_s c} f(\succ)(s, c') - \sum_{s' \in S : s' \succ_c s} f(\succ)(s', c), 0 \right\}
  \end{align*}
  Hence,
  \begin{align*}
    &\sum_{\succ\in P}\sum_{s\in S}\sum_{c\in C}\max \left\{1 - g(\succ)(s, c) - \sum_{c' \in C : c' \succ_s c} g(\succ)(s, c') - \sum_{s' \in S : s' \succ_c s} g(\succ)(s', c), 0 \right\}\\
    &\leq \sum_{\succ\in P}\sum_{s\in S}\sum_{c\in C}\max \left\{1 - f(\succ)(s, c) - \sum_{c' \in C : c' \succ_s c} f(\succ)(s, c') - \sum_{s' \in S : s' \succ_c s} f(\succ)(s', c), 0 \right\}
  \end{align*}
  
  Next, we show that the worst-case value of the objective function $B$ achievable by $f$ is also achievable by $g$. That is,
  \begin{align*}
    &\max_{\succ\in P}\sum_{s\in S}\sum_{c\in C}\max \left\{1 - g(\succ)(s, c) - \sum_{c' \in C : c' \succ_s c} g(\succ)(s, c') - \sum_{s' \in S : s' \succ_c s} g(\succ)(s', c), 0 \right\}\\
    &= \max_{\succ\in P}\sum_{s\in S}\sum_{c\in C}\max \Biggl\{\frac{1}{|A(\succ)|}\Biggl[\sum_{\pi_S\in \Pi_S,\pi_C\in \Pi_C}\biggl\{1- f\left(\left(\rho_S(\succ_{\pi_S(s)},\pi_C)\right)_{s\in S},\left(\rho_C(\succ_{\pi_C(c)},\pi_S)\right)_{c\in C}\right)\left(\pi_S(s),\pi_C(c)\right)\\ 
    &\quad -\sum_{c' \in C : c'\rho(\succ_{\pi_S(s)},\pi_C) \pi_C(c)}f\left(\left(\rho_S(\succ_{\pi_S(s)},\pi_C)\right)_{s\in S},\left(\rho_C(\succ_{\pi_C(c)},\pi_S)\right)_{c\in C}\right)\left(\pi_S(s),c'\right)\\
    &\quad -\sum_{s' \in S : s'\rho(\succ_{\pi_C(c)},\pi_S) \pi_S(s)}f\left(\left(\rho_S(\succ_{\pi_S(s)},\pi_C)\right)_{s\in S},\left(\rho_C(\succ_{\pi_C(c)},\pi_S)\right)_{c\in C}\right)\left(s',\pi_C(c)\right) \biggr\}\Biggr], 0 \Biggr\} \\
    &\leq \frac{1}{|A(\succ)|}\max_{\succ\in P}\sum_{s\in S}\sum_{c\in C}\sum_{\pi_S\in \Pi_S,\pi_C\in \Pi_C} \max \Biggl\{1- f\left(\left(\rho_S(\succ_{\pi_S(s)},\pi_C)\right)_{s\in S},\left(\rho_C(\succ_{\pi_C(c)},\pi_S)\right)_{c\in C}\right)\left(\pi_S(s),\pi_C(c)\right)\\ 
    &\quad -\sum_{c' \in C : c'\rho(\succ_{\pi_S(s)},\pi_C) \pi_C(c)}f\left(\left(\rho_S(\succ_{\pi_S(s)},\pi_C)\right)_{s\in S},\left(\rho_C(\succ_{\pi_C(c)},\pi_S)\right)_{c\in C}\right)\left(\pi_S(s),c'\right)\\
    &\quad -\sum_{s' \in S : s'\rho(\succ_{\pi_C(c)},\pi_S) \pi_S(s)}f\left(\left(\rho_S(\succ_{\pi_S(s)},\pi_C)\right)_{s\in S},\left(\rho_C(\succ_{\pi_C(c)},\pi_S)\right)_{c\in C}\right)\left(s',\pi_C(c)\right), 0 \Biggr\} \\
    &= \frac{1}{|A(\succ)|}\max_{\succ\in P}\sum_{\pi_S\in \Pi_S,\pi_C\in \Pi_C}\sum_{s\in S}\sum_{c\in C} \max \Biggl\{1- f\left(\left(\rho_S(\succ_{\pi_S(s)},\pi_C)\right)_{s\in S},\left(\rho_C(\succ_{\pi_C(c)},\pi_S)\right)_{c\in C}\right)\left(\pi_S(s),\pi_C(c)\right)\\ 
    &\quad -\sum_{c' \in C : c'\rho(\succ_{\pi_S(s)},\pi_C) \pi_C(c)}f\left(\left(\rho_S(\succ_{\pi_S(s)},\pi_C)\right)_{s\in S},\left(\rho_C(\succ_{\pi_C(c)},\pi_S)\right)_{c\in C}\right)\left(\pi_S(s),c'\right)\\
    &\quad -\sum_{s' \in S : s'\rho(\succ_{\pi_C(c)},\pi_S) \pi_S(s)}f\left(\left(\rho_S(\succ_{\pi_S(s)},\pi_C)\right)_{s\in S},\left(\rho_C(\succ_{\pi_C(c)},\pi_S)\right)_{c\in C}\right)\left(s',\pi_C(c)\right), 0 \Biggr\} \\
    &\leq \frac{1}{|A(\succ)|}\sum_{\pi_S\in \Pi_S,\pi_C\in \Pi_C}\max_{\succ\in P}\sum_{s\in S}\sum_{c\in C} \max \Biggl\{1- f\left(\left(\rho_S(\succ_{\pi_S(s)},\pi_C)\right)_{s\in S},\left(\rho_C(\succ_{\pi_C(c)},\pi_S)\right)_{c\in C}\right)\left(\pi_S(s),\pi_C(c)\right)\\ 
    &\quad -\sum_{c' \in C : c'\rho(\succ_{\pi_S(s)},\pi_C) \pi_C(c)}f\left(\left(\rho_S(\succ_{\pi_S(s)},\pi_C)\right)_{s\in S},\left(\rho_C(\succ_{\pi_C(c)},\pi_S)\right)_{c\in C}\right)\left(\pi_S(s),c'\right)\\
    &\quad -\sum_{s' \in S : s'\rho(\succ_{\pi_C(c)},\pi_S) \pi_S(s)}f\left(\left(\rho_S(\succ_{\pi_S(s)},\pi_C)\right)_{s\in S},\left(\rho_C(\succ_{\pi_C(c)},\pi_S)\right)_{c\in C}\right)\left(s',\pi_C(c)\right), 0 \Biggr\} \\
    &= \max_{\succ\in P}\sum_{s\in S}\sum_{c\in C}\max \left\{1 - f(\succ)(s, c) - \sum_{c' \in C : c' \succ_s c} f(\succ)(s, c') - \sum_{s' \in S : s' \succ_c s} f(\succ)(s', c), 0 \right\}
  \end{align*}
\end{proof}

\subsection{Proof of Theorem~\ref{thm2}}\label{app:thm2}
\begin{proof}
  Let $\succ\in P$ be given. Define a new preference profile $\succ'$ as follows. For each student $i\in\{1,\dots,n\}$ and for any two schools $j,j'\in\{1,\dots,n\}$, if in $\succ$ student $i$ prefers school $j'$ to school $j$, then in $\succ'$ we define that school $i$ prefers student $j'$ to student $j$. (A similar definition is applied when interchanging the roles of students and schools.)
  
  Let $f$ be a strategy-proof probabilistic matching mechanism. For any $\succ\in P$ and for any $i,j\in\{1,\dots,n\}$, define the mechanism $g$ by
  \begin{equation}
    g(\succ)(i,j) = \frac{1}{2}\Bigl(f(\succ)(i,j)+f(\succ')(j,i)\Bigr).
  \end{equation}
  That is, $g$ outputs the average matching probability over the two assignments given by $f(\succ)$ and by $f(\succ')$, where in the latter the roles of students and schools are interchanged. We now show that $g$ is both strategy-proof and symmetric, and that it attains the same objective function value as $f$.
  
  \textbf{(Symmetry)}\\
  We verify that $g$ satisfies symmetry. Observe that
  \begin{align*}
    g(\succ')(j,i) &= \frac{1}{2}\left(f(\succ')(j,i)+f(\succ)(i,j)\right)\\
    &= g(\succ)(i,j)
  \end{align*}
  Thus, $g$ is symmetric.
  
  \textbf{(Strategy-Proofness)}\\
  Suppose that under the true preference profile $\succ$ student $i$ submits a false preference $\triangleright_i$, and let $\triangleright$ denote the resulting profile. Next, define $\triangleright'$ in an analogous manner as follows: in $\triangleright$, for each student $i\in\{1,\dots,n\}$, if student $i$ prefers school $j'$ to school $j$, then in $\triangleright'$ we define that school $i$ prefers student $j'$ to student $j$ (and the same rule applies when interchanging the roles of students and schools). Then, under $\succ'$ if school $i$ receives the false preference $\triangleright_i$, the resulting profile coincides with $\triangleright$. Consequently, we have
  \begin{align*}
    &2\sum_{j'\succ_i j}\left(g(\succ)(i,j)-g(\triangleright)(i,j) \right) \\
    &= \sum_{j'\succ_i j}\left(f(\succ)(i,j) +f(\succ')(j,i) -f(\triangleright)(i,j) - f(\triangleright')(j,i)\right)\\
    &= \sum_{j'\succ_i j}\left(f(\succ)(i,j) - f(\triangleright)(i,j)\right) + \sum_{j'\succ_i j}\left(f(\succ')(j,i) - f(\triangleright')(j,i)\right)\\
    &\geq 0
  \end{align*}
  Thus, $g$ is strategy-proof.
  
  \textbf{(Achievablity)}\\
  We now show that the objective function value (with respect to both objective functions $A$ and $B$) achieved by $f$ is also achievable by $g$.
  
  First, consider objective function $A$. We have
  \begin{align*}
    &2\sum_{\succ\in P}\sum_{i}\sum_{j}\max \left\{1 - g(\succ)(i, j) - \sum_{j': j' \succ_i j} g(\succ)(i, j') - \sum_{i': i' \succ_j i} g(\succ)(i', j), 0 \right\}\\
    &= \sum_{\succ\in P}\sum_{i}\sum_{j}\max \Biggl\{1 - \left(f(\succ)(i, j)+f(\succ')(j,i)\right) \\
    &\quad - \sum_{j': j' \succ_i j} \left(f(\succ)(i, j')+f(\succ')(j',i)\right)- \sum_{i': i' \succ_j i} \left(f(\succ)(i', j)+f(\succ')(j,i')\right), 0 \Biggr\}\\
    &\leq \sum_{\succ\in P}\sum_{i}\sum_{j}\Biggl[\max\Bigl\{1 - f(\succ)(i, j)- \sum_{j': j' \succ_i j}f(\succ)(i, j')-\sum_{i': i' \succ_j i}f(\succ)(i', j),0\Bigr\}\\
    &\quad + \max\Bigl\{1 - f(\succ')(j, i)- \sum_{j': j' \succ_i j}f(\succ')(j', i)-\sum_{i': i' \succ_j i}f(\succ')(j, i'),0\Bigr\} \Biggr]\\
    &= \sum_{\succ\in P}\sum_{i}\sum_{j}\Biggl[\max\Bigl\{1 - f(\succ)(i, j)- \sum_{j': j' \succ_i j}f(\succ)(i, j')-\sum_{i': i' \succ_j i}f(\succ)(i', j),0\Bigr\}\\
    &\quad + \max\Bigl\{1 - f(\succ')(j, i)- \sum_{j': j' \succ'_i j}f(\succ')(j', i)-\sum_{i': i' \succ'_j i}f(\succ')(j, i'),0\Bigr\} \Biggr]\\
    &= 2\sum_{\succ\in P}\sum_{i}\sum_{j}\max \left\{1 - f(\succ)(i, j) - \sum_{j': j' \succ_i j} f(\succ)(i, j') - \sum_{i': i' \succ_j i} f(\succ)(i', j), 0 \right\}
  \end{align*}
  The last equality follows from the fact that $\succ'$ corresponds one-to-one with $\succ$. Hence, the objective function value of $f$ is achievable by $g$.
  
  Next, we show that the worst-case objective function value corresponding to objective function $B$ that is achievable by $f$ is also achievable by $g$. In particular, we have
  \begin{align*}
    &2\max_{\succ\in P}\sum_{i}\sum_{j}\max \left\{1 - g(\succ)(i, j) - \sum_{j': j' \succ_i j} g(\succ)(i, j') - \sum_{i': i' \succ_j i} g(\succ)(i', j), 0 \right\}\\
    &= \sum_{\succ\in P}\sum_{i}\sum_{j}\max \Biggl\{1 - \left(f(\succ)(i, j)+f(\succ')(j,i)\right) \\
    &\quad - \sum_{j': j' \succ_i j} \left(f(\succ)(i, j')+f(\succ')(j',i)\right)- \sum_{i': i' \succ_j i} \left(f(\succ)(i', j)+f(\succ')(j,i')\right), 0 \Biggr\}\\
    &\leq \max{\succ\in P}\sum_{i}\sum_{j}\Biggl[\max\Bigl\{1 - f(\succ)(i, j)- \sum_{j': j' \succ_i j}f(\succ)(i, j')-\sum_{i': i' \succ_j i}f(\succ)(i', j),0\Bigr\}\\
    &\quad + \max\Bigl\{1 - f(\succ')(j, i)- \sum_{j': j' \succ_i j}f(\succ')(j', i)-\sum_{i': i' \succ_j i}f(\succ')(j, i'),0\Bigr\} \Biggr]\\
    &= \max_{\succ\in P}\sum_{i}\sum_{j}\Biggl[\max\Bigl\{1 - f(\succ)(i, j)- \sum_{j': j' \succ_i j}f(\succ)(i, j')-\sum_{i': i' \succ_j i}f(\succ)(i', j),0\Bigr\}\\
    &\quad + \max\Bigl\{1 - f(\succ')(j, i)- \sum_{j': j' \succ'_i j}f(\succ')(j', i)-\sum_{i': i' \succ'_j i}f(\succ')(j, i'),0\Bigr\} \Biggr]\\
    &\leq 2\max_{\succ\in P}\sum_{i}\sum_{j}\max \left\{1 - f(\succ)(i, j) - \sum_{j': j' \succ_i j} f(\succ)(i, j') - \sum_{i': i' \succ_j i} f(\succ)(i', j), 0 \right\}
  \end{align*}
  Thus, the worst-case objective function value (i.e., the value of objective function $B$) achieved by $f$ is also achievable by $g$.
\end{proof}

\nocite{*}
\bibliography{bib}

@ARTICLE{Bade2023-ou,
  title         = "Royal processions: Incentives, efficiency and fairness in
                   two-sided matching",
  author        = "Bade, Sophie and Root, Joseph",
  journal       = "arXiv [econ.TH]",
  month         =  jan,
  year          =  2023,
  archivePrefix = "arXiv",
  primaryClass  = "econ.TH"
}

@ARTICLE{Roth1993-yk,
  title     = "Stable matchings, optimal assignments, and linear programming",
  author    = "Roth, Alvin E and Rothblum, Uriel G and Vande Vate, John H",
  journal   = "Math. Oper. Res.",
  publisher = "Institute for Operations Research and the Management Sciences
               (INFORMS)",
  volume    =  18,
  number    =  4,
  pages     = "803--828",
  month     =  nov,
  year      =  1993,
  language  = "en"
}

@ARTICLE{Abdulkadiroglu2003-ve,
  title     = "School choice: A mechanism design approach",
  author    = "Abdulkadiroğlu, Atila and Sönmez, Tayfun",
  journal   = "Am. Econ. Rev.",
  publisher = "American Economic Association",
  volume    =  93,
  number    =  3,
  pages     = "729--747",
  month     =  may,
  year      =  2003
}

@ARTICLE{Ashlagi2018-mu,
  title     = "Stable matching mechanisms are not obviously strategy-proof",
  author    = "Ashlagi, Itai and Gonczarowski, Yannai A",
  journal   = "J. Econ. Theory",
  publisher = "Elsevier BV",
  volume    =  177,
  pages     = "405--425",
  month     =  sep,
  year      =  2018,
  language  = "en"
}

@ARTICLE{Kumano2013-pl,
  title     = "Strategy-proofness and stability of the Boston mechanism: An
               almost impossibility result",
  author    = "Kumano, Taro",
  journal   = "J. Public Econ.",
  publisher = "Elsevier BV",
  volume    =  105,
  pages     = "23--29",
  month     =  sep,
  year      =  2013,
  language  = "en"
}

@ARTICLE{Gale1962-mm,
  title     = "College admissions and the stability of marriage",
  author    = "Gale, D and Shapley, L S",
  journal   = "Am. Math. Mon.",
  publisher = "JSTOR",
  volume    =  69,
  number    =  1,
  pages     =  9,
  month     =  jan,
  year      =  1962
}

@ARTICLE{Abdulkadiroglu1998-er,
  title     = "Random serial dictatorship and the core from random endowments in
               house allocation problems",
  author    = "Abdulkadiroglu, Atila and Sonmez, Tayfun",
  journal   = "Econometrica",
  publisher = "JSTOR",
  volume    =  66,
  number    =  3,
  pages     =  689,
  month     =  may,
  year      =  1998
}

@INCOLLECTION{Gale2014-ea,
  title     = "College admissions and the stability of marriage",
  author    = "Gale, David and Shapley, Lloyd S",
  editor    = "Pitici, Mircea",
  booktitle = "The Best Writing on Mathematics 2014",
  publisher = "Princeton University Press",
  address   = "Princeton",
  pages     = "299--307",
  month     =  dec,
  year      =  2014
}

@ARTICLE{Li2019-on,
  title     = "Deep Learning for Two-Sided Matching Markets",
  author    = "Li, S",
  publisher = "dash.harvard.edu",
  year      =  2019
}

@ARTICLE{Ravindranath2021-ew,
  title         = "Deep learning for two-sided matching",
  author        = "Ravindranath, Sai Srivatsa and Feng, Zhe and Li, Shira and
                   Ma, Jonathan and Kominers, Scott D and Parkes, David C",
  journal       = "arXiv [cs.GT]",
  month         =  jul,
  year          =  2021,
  archivePrefix = "arXiv",
  primaryClass  = "cs.GT"
}

@ARTICLE{Roth1982-cl,
  title     = "The economics of matching: Stability and incentives",
  author    = "Roth, Alvin E",
  journal   = "Math. Oper. Res.",
  publisher = "Institute for Operations Research and the Management Sciences
               (INFORMS)",
  volume    =  7,
  number    =  4,
  pages     = "617--628",
  month     =  nov,
  year      =  1982,
  language  = "en"
}

@ARTICLE{Aziz2019-ms,
  title     = "Random matching under priorities: stability and no envy concepts",
  author    = "Aziz, Haris and Klaus, Bettina",
  journal   = "Social Choice and Welfare",
  publisher = "Springer",
  volume    =  53,
  number    =  2,
  pages     = "213--259",
  year      =  2019
}

@ARTICLE{Von-Neumann1953-lh,
  title   = "A certain zero-sum two-person game equivalent to the optimal
             assignment problem",
  author  = "Von Neumann, John",
  journal = "Contributions to the Theory of Games",
  volume  =  2,
  number  =  0,
  pages   = "5--12",
  year    =  1953
}

@ARTICLE{Birkhoff1946-wt,
  title   = "Three observations on linear algebra",
  author  = "Birkhoff, Garrett",
  journal = "Univ. Nac. Tacuman, Rev. Ser. A",
  volume  =  5,
  pages   = "147--151",
  year    =  1946
}
\end{document}